\documentclass[11pt]{article}

\usepackage[margin=1in]{geometry}
\usepackage[T1]{fontenc}
\usepackage[utf8]{inputenc}
\usepackage{lmodern}
\usepackage{microtype}
\usepackage{amsmath,amssymb,amsthm,mathtools}
\usepackage{bm}
\usepackage{booktabs}
\usepackage{array}
\usepackage{enumitem}
\usepackage{hyperref}
\usepackage{graphicx}
\usepackage{xcolor}
\usepackage{comment} 
\excludecomment{commentg} 

\hypersetup{
    colorlinks=true,
    linkcolor=blue!60!black,
    citecolor=blue!60!black,
    urlcolor=blue!60!black
}

\newtheorem{proposition}{Proposition}
\newtheorem{theorem}{Theorem}
\newtheorem{corollary}{Corollary}
\newtheorem{remark}{Remark}

\newcommand{\R}{\mathbb{R}}
\newcommand{\E}{\mathbb{E}}
\newcommand{\norm}[1]{\left\lVert #1 \right\rVert}
\newcommand{\abs}[1]{\left| #1 \right|}

\newcommand{\1}{\mathbf{1}}
\newcommand{\diag}{\mathrm{diag}}
\newcommand{\corr}{\mathrm{corr}}

\newcommand{\depth}{\mathrm{depth}}
\newcommand{\child}{\mathrm{Ch}}
\newcommand{\parent}{\mathrm{pa}}
\newcommand{\treepath}{\mathrm{Path}}
\newcommand{\clip}{\mathrm{clip}}

\title{\textbf{Topological Risk Parity}\thanks{This document is for informational purposes only and does not constitute investment advice.}\\
}
\author{%
Revant Nayar\\{\small\textit{FMI Technologies LLC}}%
\and 
El Mehdi Ainasse\\{\small\textit{FMI Technologies LLC}}%
\and
Dnyanesh Kulkarni\\{\small\textit{FMI Technologies LLC}}%
}
\date{\today}

\begin{document}
\maketitle

\begin{abstract}
We develop \emph{Topological Risk Parity} (TRP), a tree-based portfolio construction approach intended for long/short, market neutral, factor-aware portfolios. The method is motivated by the dominance of passive/factor flows that naturally create a tree-like structure in markets. We introduce two implementation variants: (i) a rooted minimum-spanning-tree allocator, and (ii) a market/sector-anchored variant referred to here as \emph{Semi-Supervised TRP}, which imposes SPY as the root node and the 11 sector ETFs as the second layer. In both cases, the key object is a sparse rooted topology extracted from a correlation-distance graph, together with a propagation law that maps signed signals into portfolio weights.

Relative to classical Hierarchical Risk Parity (HRP), TRP is non-binary and designed for signed cross-sectional signals and hedged long-short portfolios: it preserves signal direction while using return-dependence geometry to shape exposures. It accounts for the fact that there is imperfect correlation between parent and child nodes, and thus does not propagate weights entirely to the children. We can also impose economically motivated hierarchy that involves industries, sub-industries or factors, etc. This makes it much more robust to macroeconomic shocks and crises, where within-cluster correlations might spike. These features make TRP well suited for market-neutral, equity stat-arb or L/S trend-type strategies, where enforcing neutrality or limiting exposures at the market, sector or factor level is extremely important.  

\end{abstract}

\section{Introduction}

Portfolio construction sits at the intersection of cross-sectional alpha modeling and dependence-aware risk control. Mean--variance methods in the tradition of \cite{markowitz1952} provide a canonical baseline but can be unstable in high dimensions and sensitive to covariance estimation. Network and hierarchy-based methods, including minimum-spanning-tree representations of correlation structure \cite{mantegna1999} and Hierarchical Risk Parity (HRP) \cite{lopezdeprado2016}, aim to regularize this problem by replacing dense covariance optimization with sparse or recursive structure. However, HRP based methods do not account for long-short exposures and propagate weights entirely to the child nodes.

This paper formalizes a class of algorithms that start from the same broad intuition---use the geometry of return dependence to regularize raw signals---but implement that idea in a markedly different way. The approach is based on the financial intuition that the market is organised into hierarchies of clusters and sub-clusters with distinct dynamics, with correlated dynamics introduced by the rise of ETFs and passive investing. We start with a rooted MST allocator, and we also introduce a dummy market root that forces sector ETFs to form the second layer of the hierarchy. We derive the final portfolio by multiplying the original signal of each asset by a \emph{topological factor} determined by its position in the rooted topology, followed by an $L^1$ normalization to the desired leverage.

The resulting construction is close in spirit to HRP, but differs in its core object and recursion: TRP uses a rooted sparse topology and propagates signed signal mass through that topology. The local non-conservation at branching nodes is a deliberate feature: it lets the hierarchy shape exposure while still retaining idiosyncratic positions, rather than forcing a purely conservative split at every node.

Our goal is fourfold:
\begin{enumerate}[label=(\roman*)]
    \item give a complete mathematical description of the framework;
    \item characterize its exact departure from HRP;
    \item establish basic theoretical results and bounds that explain the behavior of the method;
    \item discuss practical implications, benefits, and limitations relative to HRP.
\end{enumerate}

\section{Related background}

The present work is closest to three lines of research.

First, \emph{correlation-network methods} represent financial universes as weighted graphs, often using the Mantegna distance
\[
d_{ij} = \sqrt{\frac{1-\rho_{ij}}{2}},
\]
whose MST extracts a sparse backbone of the dependence structure \cite{mantegna1999}. This produces a graph with $n-1$ edges on $n$ vertices, emphasizing the most essential pairwise relations.

Second, \emph{hierarchical allocation methods} such as HRP construct a dendrogram from hierarchical clustering and allocate capital recursively between clusters, typically using cluster variance as the splitting criterion \cite{lopezdeprado2016}. HRP conserves capital at each split and is parity-inspired in the sense that recursive cluster allocations seek balanced risk budgets.

Third, \emph{signal overlay methods} combine a signal vector with exogenous regularization or sparse structure. TRP belongs to this family: the raw signal remains central, but it is filtered and rescaled through a rooted topological object inferred from returns.

\section{MSTs and Passive/ Factor Flows}

This section gives a simple (stylized) mechanism linking allocatable factor flows to tree-like correlation structure. The conclusion is not that markets are exactly tree-structured, but that when return co-movement is dominated by a small number of nested common components, the correlation-distance geometry becomes strongly ordered, and an MST extracts a natural sparse backbone.

\paragraph{Correlation distance and MST backbone.}
Let $\Sigma \in \R^{N\times N}$ denote the covariance matrix of returns, with correlations
\[
\rho_{ij} := \frac{\Sigma_{ij}}{\sqrt{\Sigma_{ii}\Sigma_{jj}}}.
\]
Define the Mantegna correlation distance
\begin{equation}
D_{ij} := \sqrt{\frac{1-\rho_{ij}}{2}}.
\label{eq:sec3-distance}
\end{equation}
A minimum spanning tree (MST) on the complete weighted graph with weights $D_{ij}$ selects a set of $N-1$ edges minimizing total distance. In regimes where the distances fall into well-separated tiers, the MST tends to connect within the closest tier first (basket), then the next (sector), and finally across sectors.

\paragraph{A three-level nested flow model.}
Assume each asset $i$ belongs to a sector $s(i)$ and a basket $b(i)$ (e.g., an ETF sleeve, thematic basket, or a tight industry group). Model returns as
\begin{equation}
 r_i = \theta_M z_M + \theta_S z_{s(i)} + \theta_B z_{b(i)} + \varepsilon_i,
\label{eq:sec3-flow}
\end{equation}
where $z_M$, $\{z_s\}$, $\{z_b\}$ are mutually independent shocks with unit variance, and idiosyncratic noise satisfies $\E[\varepsilon_i]=0$, $\mathrm{Var}(\varepsilon_i)=\sigma_\varepsilon^2$, and $\mathrm{Cov}(\varepsilon_i,\varepsilon_j)=0$ for $i\neq j$.

\paragraph{Implied covariance tiers.}
For $i\neq j$,
\begin{equation}
\mathrm{Cov}(r_i,r_j)=
\begin{cases}
\theta_M^2 + \theta_S^2 + \theta_B^2, & b(i)=b(j),\\
\theta_M^2 + \theta_S^2, & s(i)=s(j),\; b(i)\neq b(j),\\
\theta_M^2, & s(i)\neq s(j).
\end{cases}
\label{eq:sec3-covtiers}
\end{equation}
Moreover,
\begin{equation}
\mathrm{Var}(r_i)=\theta_M^2+\theta_S^2+\theta_B^2+\sigma_\varepsilon^2.
\label{eq:sec3-vartotal}
\end{equation}
Therefore correlations satisfy the strict ordering
\begin{equation}
\rho_{\text{same basket}} > \rho_{\text{same sector}} > \rho_{\text{different sector}},
\label{eq:sec3-corrordering}
\end{equation}
and the associated distances satisfy
\begin{equation}
D_{\text{same basket}} < D_{\text{same sector}} < D_{\text{different sector}}.
\label{eq:sec3-distordering}
\end{equation}

\paragraph{Ultrametric intuition (idealized).}
In the noiseless limit $\sigma_\varepsilon^2\to 0$ with distinct tiers, pairwise distances take only three values, determined by the deepest shared level (basket/sector/market). This is the qualitative hallmark of an ultrametric or tree-like metric, in which distances are largely determined by the lowest common ancestor.

\paragraph{Passive/ factor dominance knob.}
To make the ordering strength explicit, one may parameterize common-flow strength with $\lambda\in[0,1]$ via
\begin{equation}
 r_i = \sqrt{\lambda}\,(z_M+z_{s(i)}+z_{b(i)}) + \sqrt{1-\lambda}\,\varepsilon_i,
\label{eq:sec3-lambda}
\end{equation}
with $\mathrm{Var}(\varepsilon_i)=1$. As $\lambda\to 0$, correlations collapse and distances become harder to separate, making the MST more sample-noise-driven; as $\lambda\to 1$, the tiers in \eqref{eq:sec3-corrordering}--\eqref{eq:sec3-distordering} separate cleanly, making the MST backbone more stable.

\paragraph{Star versus deep hierarchy.}
If $\theta_M^2 \gg \theta_S^2,\theta_B^2$, then all names are dominated by the market component, correlations become nearly uniform, and the MST tends to degenerate toward a star-like structure (a single market hub). A deep hierarchy is most visible when multiple common-flow levels coexist, e.g.
\[
\theta_M^2,\theta_S^2,\theta_B^2 \gg \sigma_\varepsilon^2, \qquad \text{with } \theta_S,\theta_B \not\ll \theta_M.
\]

\paragraph{Implication for TRP.}
TRP can be viewed as a signal allocator that conditions on such a sparse backbone: it uses an MST-derived (or economically anchored) rooted tree to attenuate and redistribute signed signals in a way that is consistent with dominant co-movement structure, while remaining simpler and more robust than dense covariance optimization.

\section{Setup and notation}

Let $N$ denote the total number of assets. For each asset $i \in \{1,\dots,N\}$, let
\[
s_i \in \R
\]
be the current signed signal, and let
\[
R_i = (r_{i1},\dots,r_{iT}) \in \R^T
\]
be a history of returns or log-returns. Stack these into
\[
s = (s_1,\dots,s_N)^\top \in \R^N, \qquad
R \in \R^{N \times T}.
\]

\subsection{Activity filter}

Fix a lookback window $k \le T$, a minimum recent-magnitude threshold $\varepsilon > 0$, and a signal threshold $\tau > 0$. Define
\[
m_i := \frac{1}{k}\sum_{t=T-k+1}^T \abs{r_{it}},
\]
and the active set
\begin{equation}
\mathcal{A}
:=
\left\{
i \in \{1,\dots,N\} :
m_i > \varepsilon, \; \abs{s_i} > \tau
\right\}.
\label{eq:active-set}
\end{equation}
In our baseline specification we take $\tau = 10^{-3}$ and restrict attention to the active assets $\mathcal{A}$.

Let $n_A := |\mathcal{A}|$, and relabel the active assets as $1,\dots,n_A$ when convenient. Denote the restricted signal and return matrix by
\[
s^{(A)} \in \R^{n_A}, \qquad R^{(A)} \in \R^{n_A \times T}.
\]

\subsection{Correlation distance and the MST}

From the active return history, form the sample correlation matrix
\[
C = \corr(R^{(A)}) \in [-1,1]^{n_A \times n_A}.
\]
The implementation clips $C$ entrywise to $[-1,1]$, replaces NaNs by zero, and symmetrizes. The associated distance matrix is
\begin{equation}
D_{ij}
=
\sqrt{\frac{1 - C_{ij}}{2}}.
\label{eq:mantegna-distance}
\end{equation}
This is the standard correlation distance used in network representations of financial markets.

Construct the complete weighted graph
\[
G = (V,E,D), \qquad V = \{1,\dots,n_A\}, \quad |E| = \binom{V}{2},
\]
with edge weight $D_{ij}$. Let $\mathrm{MST}(G)$ denote a minimum spanning tree of $G$.

\begin{proposition}[Basic distance bounds]
For all $i,j$, the distance \eqref{eq:mantegna-distance} satisfies
\[
0 \le D_{ij} \le 1.
\]
Moreover, $D_{ij}=0$ if $C_{ij}=1$, and $D_{ij}=1$ if $C_{ij}=-1$.
\end{proposition}

\begin{proof}
Since $C_{ij}\in[-1,1]$, we have $(1-C_{ij})/2 \in [0,1]$, hence its square root also lies in $[0,1]$. The endpoint statements follow immediately.
\end{proof}

\begin{remark}
The MST compresses a dense $n_A(n_A-1)/2$-edge graph into a sparse backbone with exactly $n_A-1$ edges. This sparsification is one reason TRP is attractive as a regularization device: the final allocator depends only on a tree or spanning tree rather than on all pairwise interactions.
\end{remark}

\section{Topological Risk Parity (TRP)}

We now formalize two variants.

\subsection{Variant I: rooted MST allocation}

We first describe the generic rooted MST variant.

\paragraph{Step 1: active universe and MST.}
Given the active set $\mathcal{A}$, compute the MST on the distance graph of active assets.

\paragraph{Step 2: choose a root.}
Let $T = (V_T,E_T)$ be the undirected MST. Choose a root $r \in V_T$ by one of the following rules:
\begin{itemize}
    \item \textbf{hub mode}: a node of maximal degree in the MST;
    \item \textbf{max-magnitude mode}: a node with maximal $\abs{s_i}$;
    \item \textbf{fixed-index mode}: a designated index, often the first active node.
\end{itemize}
Once $r$ is chosen, orient the tree away from $r$. Because an undirected tree has a unique simple path between any two nodes, the parent/child relation induced by a fixed root is unique; any traversal can be used to recover it.

For any non-root node $v$, let $\parent(v)$ denote its (unique) parent node, and let $\child(u)$ denote the set of children of node $u$. Define the branching number
\[
b(u) := |\child(u)|.
\]

\paragraph{Step 3: mixed split-replication coefficient.}
Fix a parameter $\rho \in [0,1]$. For each internal node $u$ with $b(u)\ge 1$, define
\begin{equation}
\alpha_u(\rho)
:=
(1-\rho) + \frac{\rho}{b(u)}.
\label{eq:alpha}
\end{equation}
This is the key propagation coefficient.

\paragraph{Step 4: topological factor.}
Initialize the root with factor
\[
g_r = 1.
\]
For any non-root node $v$, define recursively
\begin{equation}
g_v
=
\alpha_{\parent(v)}(\rho)\, g_{\parent(v)}.
\label{eq:recursive-g}
\end{equation}
Equivalently, if $\treepath(r,v)$ denotes the ordered set of ancestors of $v$ on the unique path from $r$ to $v$ (including $r$ and excluding $v$ itself), then
\begin{equation}
g_v
=
\prod_{u \in \treepath(r,v)} \alpha_u(\rho).
\label{eq:path-product}
\end{equation}

\paragraph{Step 5: pre-normalization portfolio.}
The implementation uses the \emph{original} signal, not a volatility-rescaled signal:
\[
x_v = s_v^{(A)} g_v.
\]

\paragraph{Step 6: leverage normalization.}
Let $L>0$ be the target gross leverage. The final weight on active assets is
\begin{equation}
w_v
=
L \frac{x_v}{\norm{x}_1},
\qquad v \in \mathcal{A},
\label{eq:normalized-weights}
\end{equation}
provided $\norm{x}_1>0$. Inactive assets receive zero weight.

\paragraph{Remark (independence from subtree-mass exponent).}
One may compute subtree quantities
\[
S_u = \sum_{i \in \mathrm{subtree}(u)} |s_i|^p
\]
for an exponent $p \ge 1$, but in Variant I these quantities are \emph{not used} in the propagation. As a result, for fixed active set, root, and return history, the output is independent of $p$.

\subsection{Variant II: market/sector-anchored Semi-Supervised TRP}

We now describe the second variant, which we call \emph{TRP-SPY/XL}. The aim is to impose an economically meaningful prior:
\[
\text{market root} \;\rightarrow\; \text{sector ETFs} \;\rightarrow\; \text{other assets}.
\]

One can also impose an industry or sub-industry structure using ETFs or GICS labeling. 

\[
\text{market root} \;\rightarrow\; \text{sector ETFs} \;\rightarrow\;
\text{industry ETFs} \;\rightarrow\;
\text{other assets}.
\]

One can also introduce a factor-based hierarchy underneath sectors (eg. growth vs value stocks within a sector), thereby controlling factor exposures as well. In this framework, one can introduce neutrality at the sector, industry or factor level by imposing that $\sum_{i}w_{i}=0$ at that layer. 

\paragraph{Step 1: dummy market root.}
Introduce a dummy node $m$ labeled \texttt{SPY}. The dummy node is assigned signal $0$ and zero returns in the fallback branch. It does not receive portfolio weight in the final output; it exists only to organize the hierarchy.

\paragraph{Step 2: identify sector ETFs.}
Let
\[
\mathcal{X}
=
\left\{
i \in \mathcal{A} :
\text{ticker}_i \text{ begins with \texttt{XL}}
\right\}
\]
be the active set of sector ETFs. In practice this captures names such as \texttt{XLK}, \texttt{XLF}, and \texttt{XLE}.

\paragraph{Case A: sector ETFs present.}
When $\mathcal{X}\neq \varnothing$, the implementation:
\begin{enumerate}[label=(\alph*)]
    \item builds an MST on the real active assets only;
    \item augments the resulting tree by adding the dummy root $m$ and connecting $m$ to every sector ETF in $\mathcal{X}$;
    \item runs a DFS from $m$ to produce a rooted spanning tree.
\end{enumerate}

A subtle but important point is that step (b) turns a tree on $n_A$ real assets into a \emph{connected graph that generally contains cycles} whenever $|\mathcal{X}|>1$. The subsequent DFS does not merely orient a tree; it extracts a rooted \emph{spanning tree} from this connected graph.

\paragraph{Case B: no sector ETFs present.}
When $\mathcal{X}=\varnothing$, the implementation falls back to an augmented-correlation construction: the dummy root is inserted as another node and the MST is computed on the augmented universe. This is a pragmatic engineering fallback, but it is less theoretically clean because the dummy node has zero returns and therefore artificial correlations.

\paragraph{Propagation and normalization.}
Once the rooted spanning tree is formed, the same propagation law \eqref{eq:alpha}--\eqref{eq:normalized-weights} is applied to the real assets. Unlike Variant I, the Semi-Supervised TRP variant then applies post-processing:
\begin{equation}
w_i \leftarrow \clip(w_i; -c,c), \qquad
w_i \leftarrow 0 \; \text{if} \; |w_i|<\eta,
\label{eq:postprocess}
\end{equation}
for a cap $c>0$ and a minimum absolute weight threshold $\eta>0$.

\begin{remark}
The sector-anchored variant is sometimes described informally as an equal-split allocator. In the present formulation, it uses the same $\rho$-based propagation rule \eqref{eq:alpha} as Variant I; the equal-split conservative limit corresponds to setting $\rho=1$.
\end{remark}

\section{Theoretical properties}

We now collect a set of exact identities and basic theoretical results.

\subsection{Path formula and immediate consequences}

\begin{proposition}[Exact path representation]
For either TRP variant, once the rooted tree or spanning tree is fixed, each node $v$ has topological factor
\[
g_v
=
\prod_{u \in \treepath(r,v)} \left( (1-\rho) + \frac{\rho}{b(u)} \right).
\]
\end{proposition}

\begin{proof}
This follows by repeated substitution of the recursion \eqref{eq:recursive-g} along the unique parent chain from the root to $v$.
\end{proof}

\begin{proposition}[Limit cases]
For fixed rooted tree:
\begin{enumerate}[label=(\alph*)]
    \item if $\rho=0$, then $\alpha_u(0)=1$ for every internal node $u$, hence $g_v=1$ for all $v$, and the allocator reduces to normalized raw-signal allocation:
    \[
    w_v = L \frac{s_v^{(A)}}{\norm{s^{(A)}}_1};
    \]
    \item if $\rho=1$, then $\alpha_u(1)=1/b(u)$ and the allocator becomes a conservative equal-split rooted-tree method.
\end{enumerate}
\end{proposition}

\begin{proof}
Part (a) is immediate from \eqref{eq:alpha}; so is part (b).
\end{proof}

\begin{remark}
The parameter $\rho$ interpolates between \emph{no topological attenuation} ($\rho=0$) and \emph{equal-split attenuation} ($\rho=1$). This interpolation is central: topology enters continuously rather than discretely.
\end{remark}

\subsection{Exact quantification of non-conservation}

\begin{proposition}[Local mass amplification]
Let $u$ be an internal node with branching number $b(u)$. Then the sum of the immediate child factors satisfies
\begin{equation}
\sum_{v \in \child(u)} g_v
=
\beta_u(\rho)\, g_u,
\qquad
\beta_u(\rho)
:=
b(u)(1-\rho) + \rho.
\label{eq:beta}
\end{equation}
Moreover,
\[
\beta_u(\rho) \ge 1,
\]
with equality if and only if either $b(u)=1$ or $\rho=1$.
\end{proposition}

\begin{proof}
By construction, every child $v$ of $u$ satisfies
\[
g_v = \alpha_u(\rho)\, g_u
= \left((1-\rho)+\frac{\rho}{b(u)}\right)g_u.
\]
Summing over the $b(u)$ children gives
\[
\sum_{v\in\child(u)} g_v
=
b(u)\left((1-\rho)+\frac{\rho}{b(u)}\right)g_u
=
\left(b(u)(1-\rho)+\rho\right)g_u.
\]
Since
\[
\beta_u(\rho)-1 = (b(u)-1)(1-\rho)\ge 0,
\]
the claim follows.
\end{proof}

\begin{corollary}[Local expansion relative to conservative splitting]
The exact local deviation from conservative splitting is
\[
\beta_u(\rho)-1 = (b(u)-1)(1-\rho).
\]
Equivalently, whenever $\rho<1$ and $b(u)>1$, branching expands total descendant factor mass relative to the conservative benchmark $\beta_u(\rho)=1$.
\end{corollary}

\begin{corollary}[Level-mass bound]
Let
\[
\Gamma_B(\rho) := B(1-\rho)+\rho,
\]
where $B$ is the maximal branching number. If $L_\ell$ denotes the set of nodes at depth $\ell$, then
\[
\sum_{v \in L_\ell} g_v \le \Gamma_B(\rho)^\ell.
\]
\end{corollary}

\begin{proof}
At depth $0$, the only node is the root and the claim is true since $g_r=1$. Moving from depth $\ell$ to $\ell+1$, the total factor mass generated by each node is multiplied by at most $\Gamma_B(\rho)$ by \eqref{eq:beta}. Induction gives the result.
\end{proof}

\paragraph{Interpretation.}
The local factor $g_v$ along any \emph{single path} cannot exceed the ancestor factor $g_u$, yet the \emph{aggregate mass} across descendants can expand when branching is present and $\rho<1$. This is the precise sense in which the framework amplifies topology while attenuating along each individual path.

\subsection{Signal preservation and normalization}

\begin{proposition}[Sign preservation]
Suppose post-processing \eqref{eq:postprocess} is not applied. If $s_v^{(A)} \neq 0$, then
\[
\mathrm{sign}(w_v) = \mathrm{sign}(s_v^{(A)}).
\]
\end{proposition}

\begin{proof}
Since $g_v>0$ for all $v$ and $L/\norm{x}_1>0$, the sign of $w_v$ equals the sign of $x_v=s_v^{(A)}g_v$.
\end{proof}

\begin{proposition}[Scale symmetry]
Suppose post-processing \eqref{eq:postprocess} is not applied. For any scalar $c>0$,
\[
w(cs) = w(s),
\]
and for any $c<0$,
\[
w(cs) = -w(s).
\]
\end{proposition}

\begin{proof}
Because $x(cs)=c x(s)$ and the final output divides by $\norm{x(cs)}_1 = |c|\norm{x(s)}_1$, positive rescalings cancel and negative rescalings flip sign.
\end{proof}

\subsection{Sector-anchored structure in Semi-Supervised TRP}

We now formalize the one strong structural guarantee provided by the sector-anchored variant.

\begin{proposition}[Depth-one sector ETFs]
Assume the real-asset MST is connected and $\mathcal{X}\neq \varnothing$. Construct the augmented graph by adding the dummy root $m$ and edges $(m,x)$ for every $x\in\mathcal{X}$. In the DFS spanning tree rooted at $m$ produced by the implementation, every sector ETF $x\in\mathcal{X}$ is a depth-one child of $m$.
\end{proposition}

\begin{proof}
In any spanning tree rooted at $m$ that includes all edges $(m,x)$ for $x\in\mathcal{X}$ as tree edges, every sector ETF $x\in\mathcal{X}$ is a depth-one child of $m$.
\end{proof}

\begin{proposition}[Connectivity of the sector-anchored construction]
Assume the real-asset MST is connected and $\mathcal{X}\neq\varnothing$. Then the augmented graph obtained by adding the root $m$ and the edges $(m,x)$ for $x\in\mathcal{X}$ is connected, and the DFS parent map defines a spanning tree over all vertices.
\end{proposition}

\begin{proof}
The real-asset MST is connected by definition. Since $m$ is connected to every node in $\mathcal{X}$, and $\mathcal{X}$ is nonempty, $m$ is connected to the real-asset tree. Therefore the augmented graph is connected. A DFS on a connected graph always generates a spanning tree.
\end{proof}

\begin{remark}[Order dependence]
When $|\mathcal{X}|>1$, the augmented graph contains cycles. The sector ETFs remain depth-one children of the root, but the parent assignments of \emph{non-sector} assets can depend on DFS traversal order. This is a genuine implementation feature of the current code and should be acknowledged explicitly.
\end{remark}

\begin{commentg}

\section{Comparison with classical HRP}

It is essential to be precise about the relationship between TRP and classical HRP.

\subsection{What HRP does}

Classical HRP \cite{lopezdeprado2016} proceeds by:
\begin{enumerate}[label=(\roman*)]
    \item building a hierarchical clustering (dendrogram) from a distance matrix;
    \item quasi-diagonalizing the covariance matrix according to that hierarchy;
    \item recursively bisecting clusters and allocating capital between sibling clusters using inverse cluster variance.
\end{enumerate}

The core object is therefore a \emph{binary hierarchical clustering}, and the core recursion is a \emph{variance-based conservative split}. Capital is conserved at each split, and the method is parity-inspired because the recursion targets balanced risk across the hierarchy.

\subsection{What TRP does}

By contrast, TRP:
\begin{enumerate}[label=(\roman*)]
    \item builds an MST or DFS-derived spanning tree from correlation distance;
    \item chooses or imposes a root;
    \item propagates a \emph{signal-carrying topological factor} using \eqref{eq:alpha};
    \item multiplies the raw signal by that factor and normalizes.
\end{enumerate}

The core object is a \emph{rooted sparse topology}, and the core recursion acts on \emph{signed signal mass}, not on cluster variance.

\subsection{Why TRP is not HRP for \texorpdfstring{$\rho<1$}{rho<1}}

By Corollary 1, the local deviation from conservative splitting is
\[
(b(u)-1)(1-\rho).
\]
Therefore:
\begin{itemize}
    \item if $\rho<1$ and some internal node branches, the method is not locally conservative;
    \item if $\rho=1$, the method becomes conservative and equal-split, but is still not classical HRP because it does not use cluster variances.
\end{itemize}
\end{commentg}
\begin{table}[t]
\centering
\caption{Conceptual comparison between TRP and classical HRP.}
\label{tab:comparison}
\renewcommand{\arraystretch}{1.2}
\begin{tabular}{>{\raggedright\arraybackslash}p{0.22\linewidth}
                >{\raggedright\arraybackslash}p{0.24\linewidth}
                >{\raggedright\arraybackslash}p{0.24\linewidth}
                >{\raggedright\arraybackslash}p{0.24\linewidth}}
\toprule
Feature & TRP-MST & TRP-SPY/XL (Semi-Supervised TRP) & Classical HRP \\
\midrule
Primary structure & Rooted MST & Rooted spanning tree with market/sector prior & Binary dendrogram \\
Support for L/S portfolios & Yes & Yes & No \\
Root & Heuristic or fixed & Dummy market root (SPY) & No exogenous root \\
Second layer & Emergent & Sector ETFs (\texttt{XL*}) if present & Emergent from clustering \\
Propagation object & Signed signals & Signed signals & Cluster variance / risk budget \\
Local conservation & No, unless $\rho=1$ & No, unless $\rho=1$ & Yes \\
Uses original signal sign & Yes & Yes & Not central to vanilla HRP \\
Requires covariance inversion & No & No & No \\
Requires cluster variance recursion & No & No & Yes \\
Economic priors can be imposed & Limited & Strong & Limited \\
\bottomrule
\end{tabular}
\end{table}

\section{Practical implications and benefits}

\subsection{Benefits relative to HRP}

The framework offers several practical advantages.

\paragraph{1. Explicit control over hierarchy.}
The sector-anchored Semi-Supervised TRP variant allows the user to encode a market/sector/stock ontology directly. This is difficult in vanilla HRP, where the hierarchy is entirely emergent from clustering.

\paragraph{2. Support for L/S portfolios}
TRP multiplies the original signal by a positive topological factor, so the sign of the alpha is preserved prior to any clipping or thresholding. This makes the method natural for long-short overlays.
From a practitioner perspective, this matters because long--short books are typically built from signed signals and then constrained by gross leverage, net exposure, and position caps. TRP preserves the long/short intent of the signal by construction, while using the dependence topology to tilt position sizes away from crowded clusters and toward more idiosyncratic names---often reducing unintended factor bets without running a fragile covariance optimizer. In live workflows, the final weights still drop out in one pass as a vector that can be clipped, thresholded, sector-neutralized, and rescaled to the desk's risk limits.

\paragraph{3. Economic Interpretability}
TRP reflects underlying factor hierarchies such as one might see in Barra or Axioma models. While HRP only allows for binary trees, TRP captures the natural factor hierarchy that exists in markets, which involves the market branching into sectors, industries and sub-industries. One can impose or recover through the MST, different factor clusters that might exist within sectors or industries. One can then impose market, sector or factor neutrality easily.

\paragraph{4. Sparse and interpretable structure.}
The allocator depends on a tree or spanning tree rather than on the full covariance matrix. This creates an intuitive map from economic or statistical structure to final position sizes.

\paragraph{5. Smooth interpolation parameter.}
The parameter $\rho$ continuously interpolates between normalized raw-signal allocation ($\rho=0$) and equal-split rooted-tree allocation ($\rho=1$). This is a flexible control lever absent from standard HRP.

\paragraph{6. No covariance inversion and no recursive cluster variance estimation.}
Like HRP, TRP avoids covariance inversion. Unlike HRP, it does not require cluster variance calculations or quasi-diagonalization. The pipeline is comparatively simple: correlation matrix, distance matrix, MST, rooted propagation, normalization.

\paragraph{7. Natural compatibility with post-processing constraints.}
Capping, pruning of very small weights, gross leverage targeting, and sector-specific overlays are easy to append after the rooted propagation step.

\paragraph{8. Resilience to Macro Shocks}
Correlations that were low in-sample might drastically spike during macroeconomic shocks. Here the economically motivated tree saves us by ensuring that we have limited allocation to a sector, industry or factor.  



\subsection{Modeling choices and caveats}

These benefits come with modeling choices that should be understood.

\paragraph{1. Root dependence matters.}
In the generic rooted-MST variant, the root can be chosen by degree, maximal signal, or index. Different roots can change the entire allocation.

\paragraph{2. MST instability matters.}
Small changes in correlations can change the MST combinatorially. This sensitivity is shared by many graph-based methods.

\paragraph{3. A single tree discards edges.}
The MST keeps only $n_A-1$ edges, which creates interpretability and sparsity but necessarily discards information present in the full graph.

\paragraph{4. Order dependence in the sector-anchored construction.}
When multiple sector ETFs are connected to the dummy root, the augmented graph contains cycles. The resulting spanning tree for non-sector nodes can depend on traversal order.

\paragraph{5. The fallback dummy-root MST branch is pragmatic, not canonical.}
Treating a zero-return dummy root as a node in the augmented correlation matrix is a pragmatic fallback rather than a theoretically ideal construction.

\subsection{Useful structural facts}

The framework is simple enough that several properties are immediate.

\paragraph{1. Sign preservation}

As long as clipping does not zero out the name, the final weight has the same sign as the original signal:
\[
\mathrm{sign}(w_i)=\mathrm{sign}(s_i).
\]

\paragraph{2. Scale invariance of the signal}

If all signals are multiplied by the same positive constant, the final normalized weights are unchanged. If all signals are multiplied by a negative constant, all positions flip sign.

\paragraph{3. Topology only attenuates along a path}

Since
\[
\alpha_u(\rho) = (1-\rho)+\frac{\rho}{b(u)} \le 1,
\]
every additional branching level weakens the surviving conviction of downstream names. The deeper a node sits in the tree, the more structure it must pass through before becoming a final position.

\paragraph{4. The framework collapses to raw-signal allocation when topology is ignored}

At $\rho=0$, the framework simply returns a normalized signal portfolio. This provides a useful baseline and allows one to measure whether the topology is actually helping.

\begin{commentg}

\section{Discussion and extensions}

The present framework suggests several extensions.

\paragraph{1. Conservative topological allocation.}
Setting $\rho=1$ gives a conservative equal-split rooted-tree allocator. One can generalize this by replacing equal split with any normalized child score:
\[
\alpha_{u\to v}
=
\frac{q_{u\to v}}{\sum_{z\in \child(u)} q_{u\to z}},
\]
where $q_{u\to v}$ might depend on subtree risk, inverse volatility, or signal mass. This creates a broader family of conservative topological allocators.

\paragraph{2. Risk-budgeted topological propagation.}
One can replace equal split by a normalized child score $q_{u\to v}$ tied to a risk quantity (e.g., inverse subtree variance) or an exposure-control objective. This yields a family of conservative, risk-budgeted topological allocators within the same geometric pipeline.

\paragraph{3. Multi-root and multi-scale constructions.}
The sector-anchored variant already hints at a multi-root or multi-scale perspective: market root, sector layer, then stocks. A natural extension is to build explicit multiscale or multiperiod topological allocators in which the tree structure is updated more slowly than the signal layer.

\paragraph{4. More stable graph backbones.}
The MST is only one possible graph backbone. Alternatives include Planar Maximally Filtered Graphs, thresholded correlation networks, Steiner-tree approximations, or ensemble trees.
\end{commentg}
\section{Conclusion}

We began with the intuition that passive and factor flows in markets create a tree-like topological structure, captured by minimal spanning trees (MSTs). We have formalized Topological Risk Parity built around rooted sparse dependence structures captured by MSTs and signal-preserving propagation. The framework captures the essential logic of two variants: a generic rooted MST allocator and a market/sector-anchored Semi-Supervised TRP variant.

The main theoretical conclusion is simple and important. The propagation rule
\[
\alpha_u(\rho) = (1-\rho) + \frac{\rho}{b(u)}
\]
interpolates between a non-conservative regime ($\rho<1$), in which branching can expand aggregate descendant factor mass, and the conservative equal-split regime ($\rho=1$). This continuous control is the mechanism by which topology shapes exposure.

Overall, TRP offers explicit economic hierarchy, sparse interpretability, preservation of signal sign, and simple integration of practical constraints. The main sensitivities are root choice, MST sensitivity to correlations, and (in the sector-anchored construction) order dependence for non-sector descendants.

\appendix

\section{Alternative statement of the generic allocator}

For clarity, the generic TRP-MST variant can be summarized in a compact formula. Let $T_r$ be the rooted MST on the active assets with root $r$, and define
\[
g_i(T_r,\rho)
=
\prod_{u\in \treepath(r,i)}
\left(
1-\rho+\frac{\rho}{b(u)}
\right).
\]
Then the final portfolio is simply
\begin{equation}
w_i
=
L
\frac{s_i g_i(T_r,\rho)}
{\sum_{j\in\mathcal{A}} |s_j|\, g_j(T_r,\rho)}
\cdot \1_{\{i\in \mathcal{A}\}},
\label{eq:compact-final}
\end{equation}
followed, if desired, by clipping and thresholding.

\section{Recommended terminology}

To avoid conceptual confusion, the following naming convention is recommended:
\begin{itemize}
    \item \textbf{TRP-MST}: the generic rooted MST allocator;
    \item \textbf{TRP-SPY/XL} or \textbf{Semi-Supervised TRP}: the sector-anchored variant with dummy market root;
    \item \textbf{topological allocation}: the generic family name;
    \item \textbf{topological risk parity}: use only for conservative, risk-based variants, not for the generic TRP-MST construction when $\rho<1$.
\end{itemize}

\section{Additional theoretical results}

This appendix collects supplementary results that are not needed to implement TRP but may be useful for analysis.

\subsection{Bounds on topological factors}

\begin{proposition}[Bounds on topological factors]
Assume $\rho \in [0,1]$, and let
\[
B := \max_{u: b(u)\ge 1} b(u)
\]
be the maximal branching number in the rooted tree. Then for every node $v$ of depth $d_v := \depth(v)$,
\begin{equation}
\left(1-\rho + \frac{\rho}{B}\right)^{d_v}
\le
g_v
\le
1.
\label{eq:g-bounds}
\end{equation}
In particular,
\[
(1-\rho)^{d_v} \le g_v \le 1.
\]
\end{proposition}

\begin{proof}
For every internal node $u$,
\[
1-\rho + \frac{\rho}{B}
\le
1-\rho + \frac{\rho}{b(u)}
\le 1,
\]
since $b(u)\le B$ and $\rho\in[0,1]$. Multiplying along a path of length $d_v$ yields \eqref{eq:g-bounds}.
\end{proof}

\subsection{Independence from the subtree-mass exponent in Variant I}

\begin{proposition}[Independence from $p$ in Variant I]
Fix the active set $\mathcal{A}$, returns $R^{(A)}$, root $r$, and parameter $\rho$. In Variant I, the output is independent of the exponent $p$ used to compute subtree masses.
\end{proposition}

\begin{proof}
One may compute
\[
S_u = \sum_{i \in \mathrm{subtree}(u)} |s_i|^p,
\]
but in Variant I the propagation coefficients are determined solely by the branching numbers via \eqref{eq:alpha}. The subtree masses never enter the final weight computation. Therefore, once the tree itself is fixed, the output is independent of $p$.
\end{proof}

\subsection{Conditional Lipschitz stability}

The topological map is piecewise smooth: it is smooth conditional on a fixed active set, fixed correlation matrix ordering, and fixed rooted tree. Discontinuities arise only when the filter changes, when correlations cross enough to alter the MST, or when the root rule changes.

\begin{theorem}[Conditional $L^1$ stability]
Fix a rooted tree and its topological factors $g_v$. Let $D_g := \diag(g_1,\dots,g_{n_A})$, and define
\[
F(s) := L \frac{D_g s}{\norm{D_g s}_1},
\]
whenever $\norm{D_g s}_1 > 0$. If $\norm{D_g s}_1 \ge \gamma$ and $\norm{D_g s'}_1 \ge \gamma$ for some $\gamma>0$, then
\begin{equation}
\norm{F(s)-F(s')}_1
\le
\frac{2L}{\gamma}\norm{D_g(s-s')}_1
\le
\frac{2L}{\gamma}\norm{s-s'}_1.
\label{eq:lipschitz}
\end{equation}
\end{theorem}

\begin{proof}
Set $x=D_gs$ and $y=D_gs'$. Then
\[
\frac{x}{\norm{x}_1} - \frac{y}{\norm{y}_1}
=
\frac{x-y}{\norm{x}_1}
+
y\left(\frac{1}{\norm{x}_1}-\frac{1}{\norm{y}_1}\right).
\]
Taking $L^1$ norms and using $\norm{x}_1,\norm{y}_1\ge \gamma$,
\[
\norm{\frac{x}{\norm{x}_1} - \frac{y}{\norm{y}_1}}_1
\le
\frac{\norm{x-y}_1}{\gamma}
+
\frac{\big|\norm{y}_1-\norm{x}_1\big|}{\gamma}
\le
\frac{2\norm{x-y}_1}{\gamma}.
\]
Multiplying by $L$ proves the first inequality. The second follows because every $g_v\le 1$ by \eqref{eq:g-bounds}, hence $\norm{D_g z}_1 \le \norm{z}_1$ for all $z$.
\end{proof}

\begin{remark}
The theorem is intentionally \emph{conditional}. It does not claim global continuity across changes in the MST, root, or active set. Those combinatorial changes are real and should be handled empirically.
\end{remark}

\section*{Disclaimer} \textit{This document is provided for informational and educational purposes only and does not constitute investment advice, a recommendation, or an offer to buy or sell any securities, financial instruments, or investment strategies. The authors make no representations or warranties as to the accuracy, completeness, or suitability of the information contained herein. The views expressed are subject to change without notice and may not reflect the opinions of any affiliated institutions or entities.}
\end{document}